\documentclass[conference]{IEEEtran}
\IEEEoverridecommandlockouts
\usepackage{cite}
\usepackage{amsmath,amssymb,amsfonts}
\usepackage{algorithmic}
\usepackage{graphicx}
\usepackage{textcomp}
\usepackage{xcolor}
\usepackage[utf8]{inputenc}
\usepackage{amsthm}
\usepackage[margin=1in]{geometry}
\usepackage{todonotes}
\usepackage[english]{babel}
\newtheorem{theorem}{Theorem}

\let\emptyset\varnothing
\newcommand\numberthis{\addtocounter{equation}{1}\tag{\theequation}}
\usepackage{mathtools}

\newtheorem{lemma}{Lemma}
\theoremstyle{definition}
\newtheorem{definition}{Definition}
\newtheorem{assumption}{Assumption}

\def\dh{\hat{d}}

\def\cF{{\cal F}}

\def\cI{{\cal I}}

\def\cN{{\cal N}}

\def\z1{z^{-1}}

\def\BibTeX{{\rm B\kern-.05em{\sc i\kern-.025em b}\kern-.08em
    T\kern-.1667em\lower.7ex\hbox{E}\kern-.125emX}}

\begin{document}

\title{Monotonic Filtering for Distributed  Collection \\
\thanks{Zainab and Dasgupta are with the Department of Electrical and Computer Engineering, University of Iowa,
Iowa City, IA, USA,
hunza-zainab@uiowa.edu, soura-dasgupta@uiowa.edu. Audrito is with the Computer Science Department and C3S, University of Torino,
Torino, Italy,
giorgio.audrito@unito.it. Beal  is with Raytheon BBN Technologies,
Cambridge, MA, USA, jakebeal@ieee.org
Supported by the Defense Advanced Research Projects Agency (DARPA) under Contract No. HR001117C0049.
		The views, opinions, and/or findings expressed are those of the author(s) and should not be interpreted as representing the official views or policies of the Department of Defense or the U.S. Government.
		This document does not contain technology or technical data controlled under either U.S. International Traffic in Arms Regulation or U.S. Export Administration Regulations.
		Approved for public release, distribution unlimited (DARPA DISTAR case 34494, April 19, 2021).
		Dasgupta also has an appointment with the Shandong Adademy of Sciences, China.}
}

\author{Hunza Zainab, Giorgio Audrito,
 Soura Dasgupta and Jacob Beal}

\maketitle

\begin{abstract}
Distributed data collection is a fundamental task in open systems. In such networks, data is aggregated across a network to produce a single aggregated result at a source device. Though self-stabilizing, algorithms performing data collection can produce large overestimates in the transient phase. For example, in \cite{zainab:ecas} we demonstrated that in a line graph, a switch of sources after initial stabilization may produce overestimates that are quadratic in the network diameter. We also proposed monotonic filtering as a strategy for removing such large overestimates. Monotonic filtering prevents the transfer of data from device $A$ to device $B$ unless the distance estimate at $A$ is more than that at $B$ at the previous iteration. For a line graph, \cite{zainab:ecas} shows that monotonic filtering prevents quadratic overestimates. This paper analyzes monotonic filtering for an arbitrary graph topology, showing that for an $N$ device network, the largest overestimate after switching sources is at most $2N$.
\end{abstract}


\section{Introduction}

This paper proposes and analyzes a strategy called \emph{monotonic filtering}, for removing large overestimates in distributed data collection.
Recent years have witnessed a proliferation of complex networked open systems comprising a plethora of heterogeneous devices like drones, smartphones, IoT devices, and robots.
These systems  mandate the formulation of new strategies
 of collective adaptation, with the ultimate goal of transforming these environments into  \emph{pervasive computing fabric}. In these conditions, sensing, actuation, and computation are  resilient and distributed across  space \cite{bicocchi:selforganisation}.
The focus of this paper is on resilient \emph{distributed sensing}, which could be of physical environmental properties or of digital or virtual characteristics of computing resources.
By  cooperation between physically proximate, interacting sets of mobile entities, distributed sensing can support complex situation recognition \cite{coutaz:contextkey}, monitoring \cite{bicocchi:selforganisation}, and observation and control of swarms of agents \cite{viroli:aggregate:plans}.

A defining coordination task in distributed sensing is data summarization from devices in a region.
From this, one can perform many other operations like count, integrate, average, and maximize. Data summarization is like the \emph{reduce} phase of MapReduce \cite{dean:mapreduce}. It is extended to agents communicating through their neighbors and spread across a region, e.g., in wireless sensor networks \cite{talele:aggregation:survey}.
A common implementation of data summarization is by \emph{distributed collection}, where information moves towards collector devices and aggregates \emph{en route} to produce a unique result.
Such self-organizing behavior (referred to as a ``C''  block in~\cite{vabdp:selfstabilisation}), is a fundamental  and widely used component of collective adaptive systems (CASs). It can be instantiated for values of any data type with a commutative and associative aggregation operator, and can be applied in many diverse contexts.

Several papers have characterized the  dynamics of data summarization algorithms \cite{vabdp:selfstabilisation,mo:collection:error} and on  improving such dynamics \cite{ab:alp4iot:collection,abdv:aamas:collection,abdv:coord:collection}.
These papers all show that, though
 self-stabilizing, these algorithms can give rise to large transient overestimates with potentially negative consequences.
For example, if the goal is collect the net resources in a network of devices, then  overestimates, however fleeting, may cause a leader to commit to more tasks than is feasible.

For the example of a line graph, we showed in \cite{zainab:ecas} that collection can give overestimates that are quadratic in the network diameter. This is observed in face of a particular source switch after stabilization has occurred. We presented the notion of monotonic filtering as a potential amelioration. This technique prevents collection across devices whose distance towards a source or collector device is decreasing. With line graphs, \cite{zainab:ecas} showed that monotonic filtering prevents quadratic overestimates. In this paper, we analyze monotonic filtering for  general graphs  and the demonstrate that there are no quadratic overestimates in collection during the transient phase following a source switch. Rather, for an $ N $-device network the largest overestimate is at most $ 2N $.  %These results suggest that this strategy prevents excessive overestimates and can help improve dynamics of singlepath collection algorithms.

Section \ref{sprel} gives preliminaries and Section \ref{sdist} provides results on distance estimates that are used to execute monononic filtering. Section \ref{smon} defines monotonic filtering.
Section \ref{sanalysis},  proves the main result. Section \ref{ssim} gives simulations.  Section \ref{sconc} concludes.

\
\section{Preliminaries} \label{sprel}
We model a network of devices  as an undirected graph $ G=(V,E) $, with $ V=\{0,\cdots, N-1\} $ being the set of devices that are the nodes in the graph and $ E $ being the edge set of connections between devices. We assume that device $ i\in V $ carries the value $ v_i $, and that there is a designated {\em source} set of nodes $ S(t)\subset V $. The goal is to aggregate the accumulation of   the values $ v_i $ at the source set, i.e., generate accumulates $ a_i $ such that
\[ \sum_{i\in S}a_i=\sum_{i\in V}v_i.
 \]
 In particular, if there is only one source, then its accumulate should be the sum of all of the $ v_i $.
This is a special case of the data collection block (C-block) of \cite{vabdp:selfstabilisation}.

\subsection{The Basic Approach}\label{sbasic}
There are many ways to achieve this objective depending on the circumstances \cite{liu:pruteanu:dulman:gradient,adv:bisgradient,abdv:aamas:collection,abdv:coord:collection}. In this paper, we  adaptively determine a spanning tree and accumulate values from children to parents. The spanning tree  is  determined by the \emph{Adaptive Bellman-Ford} (ABF) algorithm \cite{vabdp:selfstabilisation,mo:abf:stability} a special case of what is called a $ G $-block in
\cite{vabdp:selfstabilisation}; ABF estimates distances. Two nodes are neighbors if they share an edge. Define  edge length between any two neighbors to be 1 and $ \cN(i) $ to be the set of neighbors of $ i $.
 Then  with $ \dh_i(t) $ the distance and $ S $ the source set, ABF proceeds as
\begin{align}\label{eq:distance:estimates}
	\dh_i(t) &=
	\begin{cases}
		0 & \textit{if } i \in S(t) \\
	\min_{j\in \cN(i)}\{	\dh_j(t-1)+1\} & \mbox{otherwise}
	\end{cases}.
\end{align}
 The minimizing $ j $ in the second bullet is called a current constraining or simply, constraining node, of $ i $. More precisely the constraining node $ c_i(t) $ obeys:
 \begin{align}
 c_i(t) &= \label{eq:constraining:nodes:initial}
 \begin{cases}
 i & \textit{if } i \in S(t) \\
 \underset{j\in \cN(i)}{\mbox{argmin}} \{ \dh_j(t-1)+e_{ji} \} \quad~~\, & \mbox{otherwise} \\
 \end{cases}.
 \end{align}

 These constraining nodes  set up the spanning tree: The set of children of $ i $ are the nodes they constrain:
 \begin{align*}\numberthis{}\label{eq:set:of:constraining:nodes}
 C_i(t)=\{j |~ c_j(t-1)=i\}.
 \end{align*}
 Then in keeping with the strategy outlined at the beginning  of this section, we can update the accumulate at node $ i  $ through the recursion:
 \begin{align*}\numberthis{}\label{eq:accumulated:values}
 a_i(t)=\sum_{j\in C_i(t)}a_j(t-1)+v_i.
 \end{align*}

\subsection{Quadratic Overestimates}\label{squad}
\begin{figure}[htb]
	\centering
	\includegraphics[width=\linewidth]{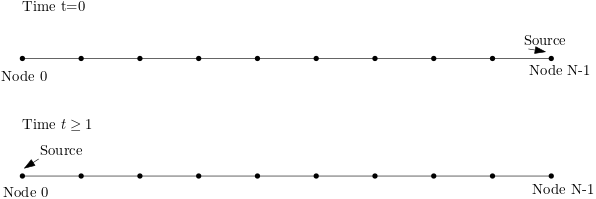}
	\caption{Representation of an $N$-node line graph ($N$-line) with a source switch from time $t=0$ to $t=1$.}\label{fig:oline:graph}
\end{figure}
From \cite{mo:abf:stability}, one knows that ABF is self-stabilizing. This means in particular that both $ c_i(t) $ and $ C_i(t) $ must acquire steady state values and thus the recursion (\ref{eq:accumulated:values}) must also converge. For example, consider the line graph in Figure \ref{fig:oline:graph}. Suppose the source is the rightmost node with index $ N-1 $, and the nodes to the left of 0 are indexed in sequence as $ 1,\cdots, N-1. $
Assume that
\begin{equation}\label{eijvi}
 v_i=1 ~\forall ~i\in V,
\end{equation}
i.e., all values are 1 and distances are hop lengths.
Then at steady state
one has
\begin{equation}\label{steady}
\dh_i(t)=N-1-i  \mbox{ and } a_i=i+1.
\end{equation}
Now suppose these values have been acquired at $ t=0 $, but the source switches from $ N-1 $ to zero at $ t=1 $. Then the following theorem from \cite{zainab:ecas} shows that \emph{en route} to self-stabilization $ a_i(t) $ suffer from quadratic overestimates in the transient phase.

\begin{theorem}\label{the:one}
	 Consider the line graph in Figure \ref{fig:oline:graph}, with (\ref{eijvi}) in force. Suppose at $ t=0 $ $ a_i(0) $ and $ \dh_i(0) $ are as in (\ref{steady}). Suppose for all $ t\geq 1 $, $ S(t)=\{0\} $. Then under (\ref{eq:distance:estimates}) and (\ref{eq:accumulated:values}), the maximum partial accumulate $a_i(t)$ reached by the source is obtained at time $t=2N-2$ and is:
	\begin{equation} \notag
	a_0(2N-2) = \left\lceil \frac{N-1}{2} \right\rceil N + N - 1 \ge \frac{N(N+1)}{2} - 1
	\end{equation}
	before reaching the correct value at time $t = 2N-1$ i.e.,
	\begin{equation}
	a_0(2N-1)=N
	\end{equation}
\end{theorem}

\subsection{Monotonic Filtering}\label{smon}
Monotonic filtering was proposed in \cite{zainab:ecas} as a remedy for such quadratic overestimation. Specifically,  this device  permits only a subset of the nodes that $ i $ constrains to be a valid children: only those $ c_i(t) $ whose distance estimate in the previous iteration was one more than $ \dh_i(t) $ are allowed to be children. Thus the set of children in (\ref{eq:set:of:constraining:nodes}) is replaced by
\begin{align}
C_i(t) = \{j| i=c_j(t-1) \wedge   \hat{d}_j(t-1)
=\hat{d}_i(t)+1\}.\label{eq:improved:C}
\end{align}
Neither the definition of constraining nodes, nor the underlying accumulation equation changes. The latter in particular remains as (\ref{eq:accumulated:values}). It should be noted that the self-stabilizing nature of ABF, per \cite{mo:abf:stability}, ensures that all distance estimates converge to their correct values. In such a steady state, under (\ref{eijvi}), every   node that constrains another automatically satisfies the  restriction on distance estimates given in (\ref{eq:improved:C}).
As proved in \cite{zainab:ecas}, this additional restriction is all that is needed to remove any overestimate from happening in a line graph. In this paper we show that this amelioration persists for general undirected graphs. The next subsection describes the analytical framework.

%restricts the set that a node constrains based on their distance estimates. The new set only adds nodes to the set constrained by a particular node if their distance estimates towards the source are decreasing.

\subsection{Definitions and Assumptions}\label{ssetting}
We now make some definitions to set up the assumptions that underlie our analysis of monotonic filtering. Suppose $ d_i $ is the hop count of $ i $ from the source set $ S $. Then from Bellman's Principle of Optimality it obeys
\begin{equation}\label{eq:true:dist:2}
\begin{cases}
d_i= 0 & i\in S\\
\min_{j\in \cN(i)} d_j+1 & \mbox{ otherwise }
\end{cases}.
\end{equation}
This leads to the definition of a \emph{true constraining node.}

\begin{definition}[True constraining node]\label{def:constraining:nodes}
	A $k$ that minimizes the right hand side of (\ref{eq:true:dist:2}) is
	a true constraining node of $i$. As there may be two
	neighbors  $j$ and $k$ of $ i $ such that $d_j =d_k$, a node may have
	multiple true constraining nodes while the true constraining node of a source is itself.
\end{definition}
We now define the notion of {\em effective diameter} introduced in \cite{mo:abf:stability}.
\begin{definition}[Effective Diameter]\label{def:def1} Consider a sequence of nodes in a graph such that each node is a true constraining node of its successor. The effective diameter D is defined as the longest length such a sequence can have in the graph.
\end{definition}
Thus, if (\ref{eijvi}) holds then the effective diameter is
\[ 1+\max_{i\in V}d_{i}. \]
We also define some important sets that are critical for our analysis.
\begin{definition}\label{dF}
	Define $ d_{ij} $ to be the minimum distance between nodes $ i $ and $ j. $ Further define $ \cF_k(m) $ to be the set of nodes whose minimum distance from node $ k $ is $ m $, i.e.,
	\[ \cF_k(m)=\{i ~|~ d_{ik}=m\}. \]
\end{definition}

If a graph $G=(V,E)$ has an effective diameter $D$ with a single source node $ 0, $ then as depicted in Figure \ref{fig:general:graph} there is a sequence of nodes, without loss of generality $i=0,1,\ldots D-1$, such that  node $i$ is the true constraining node of node $i+1$.  Henceforth we assume  the graph is as depicted in this figure. In fact, the following assumption holds.
\begin{figure}[htb]
	\centering
	\includegraphics[width=\linewidth]{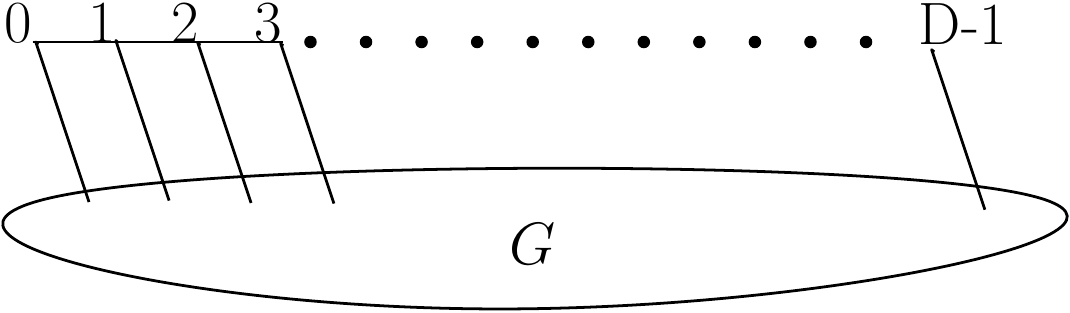}
	\caption{Arbitrary graph $G$ with a sequence highlighted.} \label{fig:general:graph}
\end{figure}
\begin{assumption}\label{ass:assumption1}
	%The values $ i $, $ v_i=1 $, $ \forall i\in V. $
	The graph $ G=(V,E) $, is as in Figure \ref{fig:general:graph}. Further (\ref{eijvi}) holds and the shortest path from $ 0 $ to $ D-1 $ has $ D $ hops.
	The algorithm used is defined by (\ref{eq:distance:estimates}, \ref{eq:constraining:nodes:initial}, \ref{eq:improved:C}).
	 For $ t\leq 0 $,  $ D-1 $ is the source, and for $ t> 0 $, $ 0 $ is the source. Further $\dh_i(0)$ and $ a_i(0) $ are steady state values of the algorithm assuming $ D-1 $ is the source.  In particular,
	\begin{equation}\label{initdhat}
	\dh_i(0)=m, ~\forall~i\in \cF_{D-1}(m).
	\end{equation}
	Also, for every integer $ m $
	\begin{equation}\label{samelevel}
	\sum_{i \in  \cF_{D-1}(m)} a_i(0)\leq N.
	\end{equation}
	Further, the effective diameter when $ 0 $ is the only source is $ D $; when $ D-1 $ is the source, the effective diameter is $D_0$.
\end{assumption}
Observe that $ D_0\geq D. $ Also note that (\ref{samelevel}) is a trivial consequence of the steady values generated by (\ref{eq:accumulated:values}) and the fact that for all $ i $, $ v_i=1 $.

\section{Evolution of distance estimates}\label{sdist}
To understand the behavior of the accumulates after the source switch, one must first understand how distance estimates evolve subsequent to the switch.
 To this end we  have the following preparatory lemma describing the true distances of neighbors from the old and new source,
using Assumption \ref{ass:assumption1}, which ensures that all distances are hop lengths.

 \begin{lemma}\label{lneighbor}
 	Suppose Assumption \ref{ass:assumption1} holds, $ D>2 $, and $ \cF_k(m) $ are as in Definition \ref{dF}. (i) Then for all $ k\in V $,  $ 0< m\leq  D-1 $ and  $ i\in \cF_k(m) $  all neighbors of $ i $ are in $ \cF_k(m-1) \bigcup \cF_k(m) \bigcup\cF_0(m+1) $.  If  $ \cF_k(m)\neq \emptyset $,  $ i $ has at least one neighbor in $ \cF_0(m-1) $.
	 (ii) Moreover, $ 0 $ and $ D-1 $ are not neighbors.

 \end{lemma}
 \begin{proof}
 	Suppose for some $ 0< m< D-1 $ and $ \ell \notin \{m-1,m,m+1\} $, $ i\in \cF_k(m) $ has a neighbor in $ \cF_k(\ell) $. If $ \ell <m-1 $ then $ i\in \cF_k(\ell+1) \neq \cF_0(m),  $ establishing a contradiction. Similarly, if $ \ell >m+1 $ then $ i\in \cF_k(\ell-1) \neq \cF_0(m),  $ again establishing a contradiction.
 	Further, the fact that $ i $ must have at least one neighbor in each of the sets  $ \cF_j(m-1) $  follows from the fact that the graph is connected.  That all neighbors of $ 0 $ and $ D-1 $ are in $ \cF_0(1) $ and $ \cF_{D-1}(1) $, respectively, follows similarly. As $ D>1 $, $ 0 $ and $ D-1 $ cannot be neighbors.

 \end{proof}

 We now provide a lemma that describes the evolution of distance estimates for $ t>0. $ To this end, we partition $ V $ into three sets $\cI_i(t)$ for $i=0,1,2$. $ \cI_0(t)$ comprises nodes that are less than $ t $ hops away from $ 0 $. In particular,
 the distance estimates at these nodes have converged to  their true distances from $ 0 $. The set $ \cI_1(t)$ comprises nodes that are not in $ I_0(t) $  but have felt the effect of a change in source. The set $ \cI_2(t) $ consists of the remaining nodes, whose distance estimates and accumulates are identical to what they were at $ t=0. $
 \begin{lemma}\label{ldist}
 	Suppose Assumption \ref{ass:assumption1} holds and $ D>2 $.
 	For all $ t\geq 1 $ consider the following partitioning  of $ V $:
 	\begin{equation}\label{I0}
 		\cI_0(t)=\bigcup_{m=0}^{t-1}\cF_0(m),
 	\end{equation}
 	\begin{equation}\label{I1}
 		\cI_1(t)=\left (\bigcup_{m=0}^{t-1}\cF_{D-1}(m)\right )\setminus \cI_0(t)
 	\end{equation}
 	and
 	\begin{equation}\label{I2}
 		\cI_2(t)=V\setminus \{\cI_0(t)\cup \cI_1(t)\},
 	\end{equation}
 	of $ V $.
 	Then under (\ref{eq:distance:estimates},\ref{eq:constraining:nodes:initial}), the following hold.
 	\begin{equation}\label{nearest}
 		\dh_i (t)=m ~\forall~i\in \cF_0 (m)\mbox{ and }  0\leq m<t,
 	\end{equation}
 	\begin{equation}\label{middle}
 		\dh_i (t)=m ~\forall~i\in \cF_{D-1} (m)\bigcap \cI_2(t).
 	\end{equation}
 	and
 	\begin{equation}\label{furthest}
 		\dh_i(t)\in \{t,t+1\} ~\forall~i\in \cI_1(t).
 	\end{equation}
 	%Further if for some $  i\in \cI_1(t) $, $ \dh_i(t)=t+1 $, then there is a $ j\in\cN(i) $ such that $ \dh_j(t)=t $.

 \end{lemma}
 \begin{proof}

 	Use induction on $ t $ to prove (\ref{nearest})-(\ref{middle}). As $ 0 $ becomes the only source at $ t=0 $, $ \dh_0(1)=0 $ verifying (\ref{nearest}).
 	From Lemma \ref{lneighbor}, 0 is not a neighbor of $ D-1 $ and $ \cN(D-1)=\cF_{D-1}(1)$. Thus from Assumption \ref{ass:assumption1}, $ \dh_i(0)=1 $ for all $ i\in \cN(D-1) $. Thus
 	$ \dh_{D-1}(1)=2, $ verifying (\ref{furthest}). %Further from Assumption \ref{ass:assumption1},  $\dh_j(1)=1 $  for all $ j\in \cN(D-1) $.

 	As  $ \cI_0(1)=\{0\}=\cF_0(0) $ and
 	$ \cI_1(1)=\{D-1\}=\cF_{D-1}(0) $,  $ \cI_2(1)=V\setminus \{0,D-1\}. $ Thus
 	\begin{equation}\label{key}
 		\cI_2(1)=\left  (\bigcup_{m=1}^{D_1-1}\cF_{D-1}(m)\right )\setminus \{0\}.
 	\end{equation}
 	From Assumption \ref{ass:assumption1}, $ \dh_i(0)=m $ for all $ i\in \cF_{D-1}(m) $.  Thus from (\ref{eq:distance:estimates},\ref{eq:constraining:nodes:initial}), (\ref{middle}) holds.

 	Now suppose the result holds for some $ t=T \geq 1. $
 	By the induction hypothesis for
 	\begin{equation}\label{2T-1nearest}
 		\dh_i(T)=m<T ~\forall~i\in \cF_0 (m)\mbox{ and }  1\leq m<T.
 	\end{equation}
 	Further
 	\begin{equation}\label{2T-1midd}
 		\dh_i(T)=m    ~\forall~ i\in \cF_{D-1} (m)\bigcap\cI_2(T).
 	\end{equation}
 	Lastly
 	\begin{equation}\label{2T-1furthest}
 		\dh_i(T)\in \{T,T+1\} ~\forall~  i\in  \cI_1(T).
 	\end{equation}

 	\vspace{2mm}
 	\noindent
 	{\bf Proof of (\ref{nearest}):}
 	Suppose now $$ i\in \cI_0(T+1)=\bigcup_{m=0}^T \cF_0(m), $$ in particular for some $ 0\leq m\leq T $, $ i\in \cF_0(m). $ From (i) of Lemma \ref{lneighbor},
 	\[ \cN(i)\subset \bigcup_{k=m-1}^{m+1}\cF_0(k) \]
 	with at least one $ j\in \cN(i)\bigcap \cF_0(m-1). $ From (\ref{2T-1nearest}), $ \dh_j(T)=m-1. $ Thus $ \dh_j(T+1)\leq m. $  To establish a contradiction, suppose for some other $ n\in \cN(i) $,
 	\begin{equation}\label{ncontra}
 		\dh_n(T)<m-1\leq T-1.
 	\end{equation}
 	From (\ref{2T-1furthest}),   $ n\notin \cI_1(T)  $.
 	Further, from (i) of Lemma \ref{lneighbor}, $ n\notin \cF_0(\ell) $, for $ \ell<m-1. $ Thus from (\ref{2T-1nearest}), $ n\notin \cI_0(T) $ and
 	$ n $ must be in $  \cI_2(T). $ From (\ref{middle})
 	\[ \dh_n(T)=d_{n,D-1}\geq T \]
 	violating (\ref{ncontra}). Thus $ \min_{k\in \cN(i)}\{\dh_k(T)\}=m-1 $ and from (\ref{eq:distance:estimates},\ref{eq:constraining:nodes:initial}), $ \dh_i(T+1)=m $, proving (\ref{nearest}).

 	\vspace{2mm}
 	\noindent
 	{\bf Proof of (\ref{furthest}):}
 	Now consider
 	\begin{equation}\label{indI1}
 		i\in \cI_1(T+1)=\left (\bigcup_{m=0}^{T}\cF_{D-1}(m)\right )\setminus \cI_0(T+1).
 	\end{equation}
 	%and $ i\in \cI_2(T+1). $
 	To establish a contradiction, first assume that $ \dh_i(T+1)>T+2. $ This implies
 	\[ \min_{k\in \cN(i)}\{\dh_k(T)\}>T+1. \]
 	From (\ref{2T-1midd}), no neighbor of $ i $ is in $  \cI_1(T) $.
 	As from
 	(i) of Lemma \ref{lneighbor} and (\ref{indI1}), $ i $ must have at least one neighbor in
 	\[ \bigcup_{m=0}^{T-1} \cF_{D-1}(m), \]
 	Because of (\ref{furthest}) and (\ref{I2}) any
 	\[j\in \cN(i)   \bigcap \left \{\bigcup_{m=0}^{T-1} \cF_{D-1}(m)\right \}\]
 	must be in $ \cI_0(T) $ and additionally $ \dh_j(T)>T+1 $. This violates (\ref{nearest}), i.e., the upper bound of (\ref{furthest}) holds.

 	Thus, to prove (\ref{furthest}) for $ t=T+1 $ we need to show that
 	\begin{equation}\label{1minbd}
 		\dh_i(T+1)\geq T.
 	\end{equation}
 	To establish a contradiction suppose (\ref{1minbd}) is false. Then
 	\begin{equation}\label{thereexists}
 		\exists ~j\in\cN(i) \mbox{ such that }  \dh_j(T)<T.
 	\end{equation}
 	In view of (\ref{2T-1furthest}), $ j\notin \cI_1(T) $.
 	Thus suppose $ j\in\cI_0(T) $.  Then because of (\ref{nearest}), $ j\in \cF_0(\ell) $ for some $ \ell<T. $ As $ i\in \cN(j) $ this means $ i\in \cF_0(m) $ for some $ m\leq T. $ Thus from (\ref{I0}), $ i\in  \cI_0(T+1) $. From (\ref{I1}),  $ i\notin \cI_1(T+1)$, establishing a contradiction.

 	Thus $ j\in \cI_2(T). $ From (\ref{I2}), for some $ m\geq T $,
 	$ j\in \cF_{D-1}(m)\bigcap\cI_2(T). $ Thus from (\ref{2T-1midd}), $ \dh_j(T)\geq T $, violating (\ref{thereexists}). Thus (\ref{furthest}) holds for $ t=T+1. $

 	\vspace{2mm}
 	\noindent
 	{\bf Proof of (\ref{middle}):}
 	Finally consider $ i\in \cI_2(T+1) $, i.e., from
 	(\ref{I1}, \ref{I2})
 	\begin{equation}\label{2indI}
 		i\in \cF_{D-1}(m), \mbox{ for some } m>T.
 	\end{equation}
 	From (i) of Lemma \ref{lneighbor} there exists
 	\begin{equation}\label{jnew}
 		j\in \cN(i)\bigcap \cF_{D-1}(m-1).
 	\end{equation}
 	As $ m-1>T-1 $, from (\ref{I1})  $ j\notin \cI_1(T) $. Suppose $ j\in \cI_0(T) $, then from (\ref{2T-1nearest})
 	\[
 	\dh_j(T) \leq T-1<m-1. \]
 	If $ j\in \cI_2(T) $, then from (\ref{2T-1midd}) $ \dh_j(T)=m-1 $. Thus
 	$ \dh_i(T+1)\leq m. $

 	Thus the violation of  (\ref{middle}) implies
 	\begin{equation}\label{2contra}
 		\dh_j(T)<m.
 	\end{equation}
 	Then we assert that
 	\begin{equation}\label{notI0}
 		j\notin \cI_0(T).
 	\end{equation}
 	To establish a contradiction, suppose (\ref{2contra}) holds but (\ref{notI0}) does not.
 	Then from (\ref{2T-1nearest}),  $ j\in \cF_0(\ell) $ for some
 	$ \ell<T. $ As $ i\in\cN(j) $ from (i) of Lemma \ref{lneighbor}, $ i\in \cF_0(k) $, with $ k\leq T $, i.e., $ i\in \cI_0(T+1) $. This leads to a contradiction as $ \cI_1(T+1) $ and $ \cI_0(T+1) $ and $ i\in\cI_1(T+1) $  are disjoint.bThus (\ref{notI0}) holds. Consequently, $ j\in \cI_2(T) $ and $ \dh_j(T)=m-1 $ and
 	$ \dh_i(T+1)= m $, proving (\ref{middle}).
 \end{proof}

\section{Accumulates Under Monotonic Filtering} \label{sanalysis}
We now compute the partial accumulates under the new definition of children given in (\ref{eq:improved:C}), using the distance estimates characterized in Lemma \ref{ldist}, and the evolution of accumulates defined in  (\ref{eq:accumulated:values}). The added constraint in (\ref{eq:improved:C}), which intuitively ensures that data is collected by always descending distances, is satisfied by every node in a stable state; however, it may not be satisfied during transients. We will show that monotonic filtering suffices to  eliminate quadratic overestimates in a general graph.
%\todo[inline]{"Completely eliminate" contradicts our assertions at the start of the paper.}

We consider the partial accumulates in each of our partitioned sets individually. The following lemma characterizes the accumulates in $\cI_1(t)$.

\begin{lemma}\label{lem:a_i at I_1}
	Under Assumption \ref{ass:assumption1},  the partial accumulates in  $\cI_1(t)$ defined in (\ref{I1}), obey
	\begin{equation}\label{eq:lem4}
		a_i(t)=1 \text{  }\forall \text{  } i \in \cI_1(t)
	\end{equation}
\end{lemma}
\begin{proof}
	Our new definition of the set $C_i(t)$ given in (\ref{eq:improved:C}) requires that $C_i(t)$ accept a node $j\in \cN(i)$ only if $\hat{d}_j(t-1)=\hat{d}_i(t)+1$.
	The partial accumulates are given in (\ref{eq:accumulated:values}) with $v_i=1$ as
	\begin{equation}\label{eq:ais_improvedC}
		a_i(t)=\sum_{j \in C_i(t)} a_j(t-1)+1
	\end{equation}
	To prove (\ref{eq:lem4}), we need to show that $C_i(t) = \emptyset$ for all $i \in \cI_1(t)$. To prove this, it suffices to show that $\hat{d}_j(t-1)\neq \hat{d}_i(t)+1$ for all $j \in \cN(i).$

	From (\ref{furthest}) in Lemma \ref{ldist}, we know that the distance estimates for nodes in $i \in \cI_1(t)$ satisfy
	\begin{equation}\label{eq:lem4proof1}
		\hat{d}_i(t)+1 \in \{t+1,t+2\}
	\end{equation}
	Thus, if $\hat{d}_j(t-1) \notin \{t+1,t+2\}$, then $j \notin C_i(t)$. Consider cases depending on whether a neighbour $j \in \cN(i)$ is in sets $\cI_0(t-1)$, $\cI_1(t-1)$ and $\cI_2(t-1)$.

	\textbf{Case I $ j \in \cI_1(t-1) $:} Then $\hat{d}_j(t-1)\in \{t-1,t\} \notin \{t+1,t+2\}$ from (\ref{furthest}) in Lemma \ref{ldist}. Hence $j \notin C_i(t)$.

	\textbf{Case II $ j \in \cI_0(t-1) $:} Then $\hat{d}_j(t-1)= d_{0j} \notin \{t+1,t+2\}$ as $0 \leq d_{0j} <t-1$ from (\ref{nearest}) in Lemma \ref{ldist}. Hence $j \notin C_i(t)$.

	\textbf{Case III $ j \in \cI_2(t-1) $:} The nodes contained in the set $\cI_2(t-1)$ are given by (\ref{I2}), that is:
	\begin{equation}\label{eq:lem4c31}
		\cI_2(t-1)= V \setminus \{\cI_0(t-1) \cup \cI_1(t-1)\}
	\end{equation}
	The distance estimates for these nodes  from Lemma \ref{ldist} and (\ref{middle}) are as follows:
	\begin{equation}\label{eq:lem4case31}
		\hat{d}_j(t-1) = d_{D-1, j}.
	\end{equation}
	As $i \in \cI_1(t)$, then $d_{D-1, i} \le t-1$ by (\ref{I1}).  For $j$ to be a neighbor of a node in set $\cI_1(t)$, $d_{D-1, j} \leq  d_{D-1, i}+1 \leq t$ must hold. But $t  \notin \{t+1,t+2\}$ and hence $j \notin C_i(t)$ concluding the proof.
\end{proof}

To characterize the accumulates of nodes in $\cI_0(t)$, we need information on how accumulates evolve in the neighboring nodes that are part of the set $\cI_2(t)$. The following lemma gives an upper bound on that.

\begin{lemma}\label{lem:lem5:distI0nI2}
	Under Assumption \ref{ass:assumption1},   nodes $j \in C_i(t) \bigcap \cI_2(t)  $ for all $i \in \cI_0(t)$ obey:
	\begin{equation}\label{eq:lem5dist}
		\sum_{i \in \cI_0(t)} \sum_{j \in C_i(t) \bigcap \cI_2(t) } a_j(t-1) \leq N
	\end{equation}
\end{lemma}
\begin{proof}
	Suppose $ i\in \cI_0(t) $ has a neighbor  $ j\in \cI_2(t). $  Then   $ d_{j,0}\geq t. $ Thus from Lemma \ref{lneighbor}, $ d_{i,0}=t-1 $. From Lemma \ref{ldist}, $ \dh_i(t)=t-1 $. Thus,
	\[ \dh_j(t)=t-1+1=t, \forall i \in \cI_0(t) \mbox{ and }  j \in C_i(t) \bigcap \cI_2(t). \]
	Thus from Lemma \ref{ldist} all such $ j $ have the same distance from $ D-1. $ As accumulates at these $ j $ are the value they had at the steady state at $ t=0 $, the result follows from (\ref{samelevel}).
\end{proof}

Given the information regarding neighbors of $\cI_0(t)$ in other sets, we can now characterize the accumulates in $\cI_0(t)$ according to the following lemma.

\begin{lemma}\label{lem:a_i at I_0}
	Under Assumption \ref{ass:assumption1}, the partial accumulates at  nodes in the set $\cI_0(t)$ obey
	\begin{equation}\label{eq:lem5}
		a_i(t) \leq 2N  \text{  }\forall  \text{  } i \in \cI_0(t)
	\end{equation}
\end{lemma}

\begin{proof}
		%From Lemma \ref{ldist}, $ \cI_0(1)=\{0\} $ and $ \dh_1(0)=0. $ As distance estimates are nonnegative $ C_0(1)=\emptyset $ and from (\ref{eq:accumulated:values}) $ a_0(1) =1$ and (\ref{eq:lem5}) holds for $ t=1. $  Similarly, for all $ i\in \cI_1(2) $, $ \dh_i(1)\leq 1. $ As only $ 0 $ has a distance estimate of zro at $ t=2 $,
		%$ a_i(2) \leq 2$ for all $ i\in \cI_0(2). $ Thus the result holds for $ t\in \{1,2\}. $

	The result holds for all $ t>0 $, if the sum of all partial accumulates at a given distance $ x < t$ from new source node $0$ is bounded by the following,
	\begin{equation}\label{eq:lem6prop}
		\sum_{i \in \cF_0(x)} a_i(t) \leq N + \sum_{y=x}^{D-1} |\cF_0(y)|.
	\end{equation}
	%(\ref{eq:lem5}) will follow.

	We  prove (\ref{eq:lem6prop}) by induction on  decreasing $x$.

	%\textbf{Base case:}
	%Take t=0, then the only node in $\cI_0(t)$ is the new source node $0$ and it is at a distance $x=0$. The set $C_0(0)$ is empty and hence $a_0(0)=1$.

	To initiate the induction, for any  consider $x = t-1$ and $ x=t-2 $ in the latter case provided $ t\geq 2. $
	We know from (\ref{eq:accumulated:values}) that
	\begin{equation}\label{eq:lem6indhyp1}
		a_i(t)= \sum_{j \in C_i(t)} a_j(t-1)+1
	\end{equation}
	Then taking a sum on the left hand side of (\ref{eq:lem6prop}), we have
	\begin{align}
		\sum_{i \in \cF_0(x)}a_i(t)=  \sum_{i \in \cF_0(x)}\left  ( \sum_{j \in C_i(t)} a_j(t-1)+1 \right )\label{eq:lem6indhypess}\\
		=  \sum_{i \in \cF_0(x)} \sum_{j \in C_i(t)} a_j(t-1)+|\cF_0(x)|\label{eq:lem6indhypess1}
	\end{align}
	Since $x = t-1$ or $t-2$,  $j \in C_i(t)$ cannot be in $\cI_0(t-1)$. If $j \in C_i(t)$ are in  $\cI_2(t-1)$ then from Lemma \ref{lem:lem5:distI0nI2}
	\begin{align*}
		\sum_{i \in \cF_0(x)} \sum_{j \in C_i(t) \bigcap \cI_2(t-1)} a_j(t-1) \\
		\leq \sum_{i \in \cI_0(t)} \sum_{j \in C_i(t) \bigcap \cI_2(t-1) } a_j(t-1) \leq N
	\end{align*}
	The contribution from $\cI_1(t-1)$ excludes all nodes from $\cI_0(t-1)$, and everything that is a distance $t-1$ or less away from source node $0$ is in $\cI_0(t-1)$. It follows that the total contribution from nodes in $\cI_1(t-1)$ must be
	\begin{equation}\label{eq:lem6:jinI1}
		\sum_{i \in \cF_0(x)} \sum_{j \in C_i(t) \bigcap \cI_1(t-1)} \!\! a_j(t-1) \leq \sum_{y=x+1}^{D-1} |\cF_0(y)|.
	\end{equation}
	Thus (\ref{eq:lem6indhypess1}) simplifies to
	\begin{align*}\label{eq:lem6indhypess121}\numberthis
		\sum_{i \in \cF_0(x)}a_i(t) \leq N+\sum_{y=x+1}^{D-1} |\cF_0(y)| + |\cF_0(x)|  \\
		\leq N+\sum_{y=x}^{D-1} |\cF_0(y)|.
	\end{align*}

	Now suppose (\ref{eq:lem6prop}) holds for all $ t \geq 3 $.
	For all other $x<t-3$ for any $t$, the set $C_i(t) \subseteq \cF_0(x+1) \subseteq \cI_0(t-1)$). The double sum in (\ref{eq:lem6indhypess1}) simplifies to
	\begin{equation}\label{eq:lem6simp}
		\sum_{i \in \cF_0(x)} \sum_{j \in C_i(t)} a_j(t-1) = \sum_{j \in \cF_0(x+1)} a_j(t-1)
	\end{equation}
	The inductive hypothesis in (\ref{eq:lem6prop}) can be applied directly to (\ref{eq:lem6simp}) in (\ref{eq:lem6indhypess1}) as
	\begin{align*}\label{eq:lem6:t2}\numberthis
		\sum_{i \in \cF_0(x)}a_i(t)=  \sum_{j \in \cF_0(x+1)} a_j(t-1)+|\cF_0(x)| \\
		 \leq N + \sum_{y=x+1}^{D-1} |\cF_0(y)| +|\cF_0(x)|.
	\end{align*}

	Then the following holds
	\begin{equation}\label{eq:lem6inhypcase322ineq}
		\sum_{i \in  \cF_0(x)}a_i(t) \leq N + \sum_{y=x}^{D-1} |\cF_0(y)|
	\end{equation}
	completing the proof of (\ref{eq:lem6prop}).
\end{proof}

We now have the main result.

\begin{theorem}\label{the:1}
	In a general graph, with monotonic filtering the partial accumulations $a_i(t)$ in sets $\cI_0(t)$, $\cI_1(t)$ and $\cI_2(t)$ for any time $t$, have an upper bound that is given as:
	\begin{equation}\label{eq:th:11}
		a_i(t) \leq 2N \text{ } \forall  \text{ } i \in \cI_0(t)
	\end{equation}
	\begin{equation}\label{eq:th:12}
		a_i(t) =1  \text{ } \forall  \text{ } i \in \cI_1(t)
	\end{equation}
	and
	\begin{equation}\label{eq:th:13}
		a_i(t) \leq N \text{ } \forall  \text{ } i \in \cI_2(t)
	\end{equation}

\end{theorem}
\begin{proof}
	The first two expressions follow directly from lemma \ref{lem:a_i at I_0} and \ref{lem:a_i at I_1} respectively.

	The proof of (\ref{eq:th:13}) is trivial as nodes in the set $\cI_2(t)$ have not been affected yet by the source switch. They maintain partial accumulate values attained at the previous steady state, and these values never exceed $N$, the total number of nodes in the graph.
\end{proof}

\section{Simulation}\label{ssim}

In order to evaluate the performance of monotonic filtering in a more general setting, we simulated a network of 100 to 1000 nodes, randomly displaced in a square, with a connection range such that the average number of neighbours per node is about 10. We performed 1000 simulations and averaged the results. The simulation is publicly available online.\footnote{https://github.com/Harniver/monotonic-filtering-dynamics}

\begin{figure}[htbp]
	\centering
	\includegraphics[width=\linewidth]{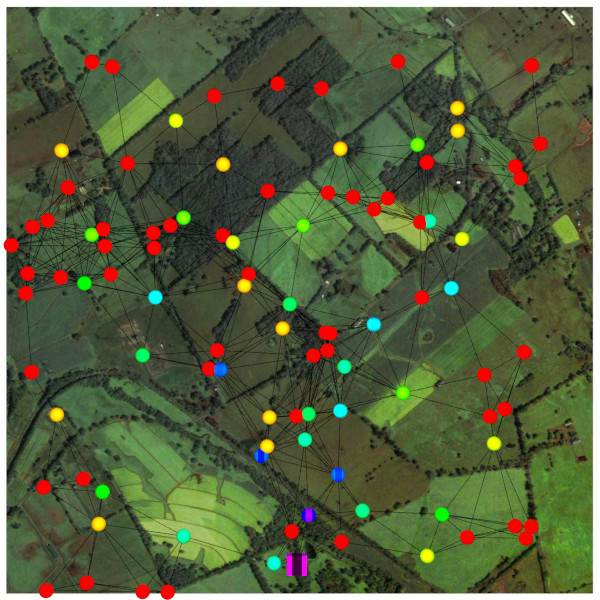}
	\caption{Simulation of collection algorithms on a random bidimensional arrangement. Colors of nodes are tuned from red (estimate 1) to magenta (correct estimate) to black (infinitely large estimate); and the central color corresponds to the basic approach while the color of the sides corresponds to monotonic filtering. The screenshot is taken during a transient recovery, in which the source (big square) has a very large overestimate with the basic approach (black central color) while it is almost correct with monotonic filtering (magenta sides).}
	\label{fig:screenshot}
\end{figure}
\begin{figure*}[h]tb
	\centering
	\includegraphics[width=0.49\linewidth]{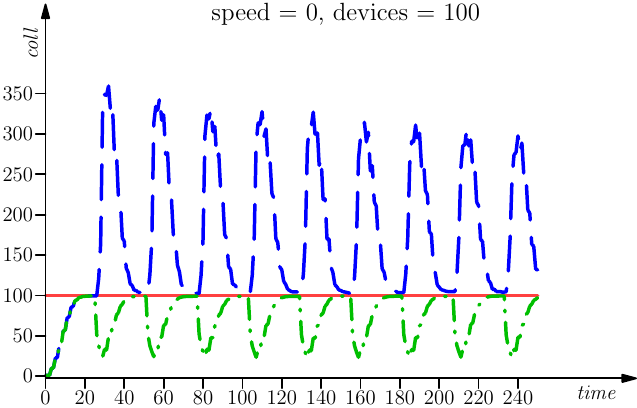}
	\includegraphics[width=0.49\linewidth]{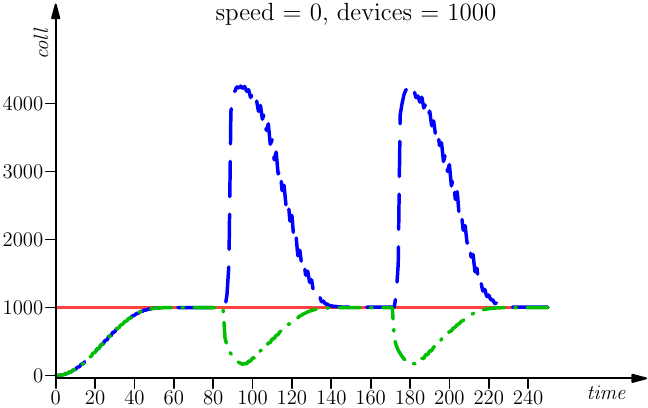}
	\\[5pt]
	\includegraphics[width=0.49\linewidth]{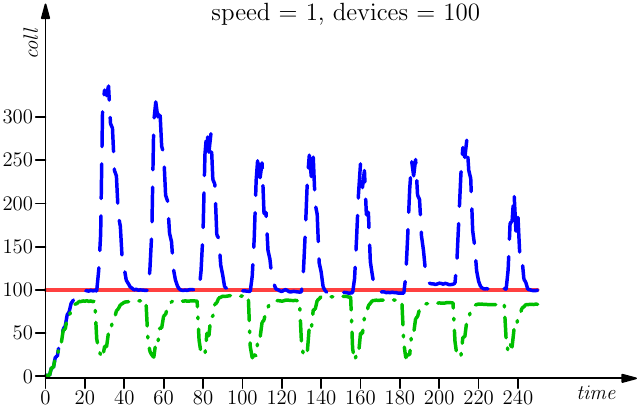}
	\includegraphics[width=0.49\linewidth]{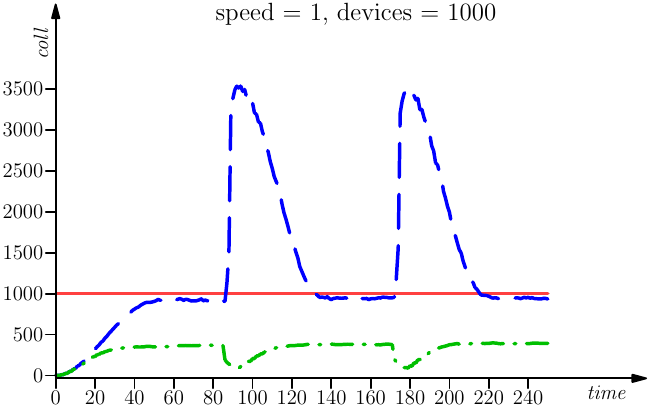}
	\\[5pt]
	\includegraphics[width=0.49\linewidth]{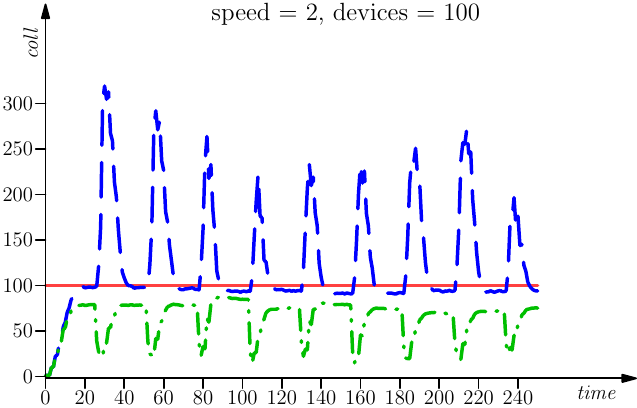}
	\includegraphics[width=0.49\linewidth]{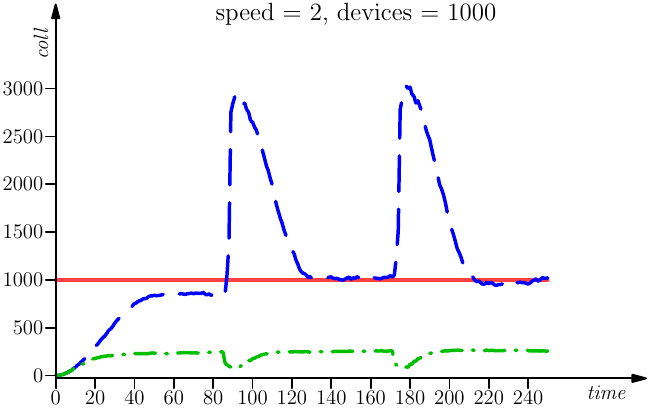}
	\\[5pt]
	\includegraphics[scale=1.5]{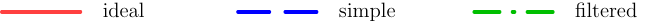}
	\caption{Average collection results (counting device number) in source nodes over simulated time, where the source is periodically switching to different devices (at times matching the periodic spikes). The lines correspond to the number of devices (\emph{ideal}, in red), to the basic approach (\emph{simple}, in blue), and to monotonic filtering (\emph{filtered}, in green). The basic approach suffers from significant overestimates at every source change, while monotonic filtering provides underestimates instead. In mobile networks (\emph{speed} greater than zero) with a sufficient number of devices, the underestimates given by monotonic filtering become systematic, while the basic approach is still able to reach a reasonably correct value.}
	\label{fig:plots}
\end{figure*}

A sample screenshot of the simulation is depicted in Figure \ref{fig:screenshot}, while synthetic plots of the average collection results in source nodes are given in Figure \ref{fig:plots}.
When the position of the nodes is fixed (\emph{speed} is zero), monotonic filtering can prevent any overestimate from occurring.
The estimate still converges to the correct result as fast as with the basic approach, which instead suffers from high peaks during reconfiguration. When nodes are steadily moving (\emph{speed} is positive), both algorithms start underestimating the true count as the speed increases and number of devices increase. However, the underestimates are much more pronounced with monotonic filtering, while the basic approach is able to tolerate small speeds.

\section{Conclusion}\label{sconc}

In this paper, we investigated the effect of a monotonic filtering condition on the transient values of a single-path collection algorithm during recovery from a source switch. In a static graph, the monotonic filtering condition is proved to bound overestimates to at most $2N$, while single-path collection without it is showed to reach quadratic overestimates in some cases. By evaluating the algorithms in simulation, we show that in practice transient large overestimates do occur without filtering, while no overestimate at all is present with monotonic filtering. Finally, we also simulate the behavior of the algorithms under persistent perturbations, i.e., steady movement of the nodes. In this scenario, both algorithms degrade their quality towards underestimates; however, with monotonic filtering the degradation occurs much sooner, with lower movement speeds and lower number of devices.

In future work, monotonic filtering should be compared with existing strategies \cite{abdv:list:collection}, in its ability to avoid overestimates under persistent perturbations.

\bibliographystyle{IEEEtran}
\bibliography{IEEEabrv,bibliograpy}
\end{document}